\documentclass[11pt]{article}

\usepackage{epsfig,amsmath,latexsym,amssymb}
\usepackage{graphicx}
\usepackage{lscape}
\usepackage{picture, eso-pic, tikz} 
\usepackage{dsfont}
\usepackage{listings}
\usepackage{changes}
\usepackage{bbding}
\usepackage{url}
\usepackage{hyperref}
\usepackage{tablefootnote}
\usepackage{booktabs}

\usepackage{tikz}
\usetikzlibrary{arrows,shapes,er,positioning,arrows.meta,calc}
\tikzset{
     arrow/.style = { thick,  ->, >=Triangle},
}

\renewcommand{\baselinestretch}{1.1}
\def\E{{\mathbb E}}  %
\def\Beweis{\footnotesize}

\newcommand{\Remm}[1]{}
\newtheorem{theo}{Theorem}[section]

\newtheorem{prop}[theo]{Proposition}

\newtheorem{model ass}[theo]{Model Assumptions}

\newtheorem{rems}[theo]{Remarks}

\def\EndProof{{\begin{flushright}\vspace{-2mm}$\Box$\end{flushright}}}

\numberwithin{equation}{section}

\definecolor{MyGray}{rgb}{0.92,0.92,0.92}


\definecolor{British racing}{rgb}{0.0, 0.5, 0.0}

\def\bX{\boldsymbol{X}}

\def\bz{\boldsymbol{z}}

\begin{document}
\author{Ronald Richman\footnote{insureAI, ronaldrichman@gmail.com}
  \and
  Mario V.~W\"uthrich\footnote{Department of Mathematics, ETH Zurich,
mario.wuethrich@math.ethz.ch}}

\date{Revised Version of \today}
\title{From Chain-Ladder to Individual Claims Reserving}
\maketitle

\begin{abstract}
\noindent	
The chain-ladder (CL) method is the most widely used claims reserving technique in non-life insurance. This manuscript introduces a novel approach to computing the CL reserves based on a fundamental restructuring of the data utilization for the CL prediction procedure. Instead of rolling forward the cumulative claims with estimated CL factors, we estimate multi-period factors that project the latest observations directly to the ultimate claims. This alternative perspective on CL reserving creates a natural pathway for the application of machine learning techniques to individual claims reserving. As a proof of concept, we present a small-scale real data application employing neural networks for individual claims reserving.
\medskip


\end{abstract}

\setcounter{section}{-1}

\section{Addendum}
After uploading the first version of this manuscript to {\tt arXiv}, Florian Gerhardt (thank you very much!) pointed out that our main result had already been established on page 130 of Lorenz--Schmidt \cite{LorenzSchmidt}, where the approach is referred to as the grossing up method. Following the theorem in that reference, the authors remark that “..., the grossing up method is irrelevant in practice.” We believe that this assessment is too pessimistic in the era of machine learning. In our view, precisely this structural perspective may provide the key to bringing individual claims reserving into practical application. For this reason, we have chosen to leave the paper essentially unchanged, in the hope of convincing the reader that this approach offers a promising direction for future research.

\section{Introduction}
About a decade ago, research of individual claims reserving utilizing machine learning (ML) techniques began to emerge. Since then, numerous methods and models have been proposed, including regression trees, gradient boosting machines, and neural networks. Nevertheless, the area of individual claims reserving remains predominantly a research domain and has not yet achieved a widespread adoption in industry practice. Schneider--Schwab \cite{Schwab} write: ``Typically, newer models which consider richer data on individual claims are either parametric or use machine learning techniques. However, none have become a gold standard, and advances are still needed.''
We attribute this fact to various challenges. First, it is difficult to find publicly available individual claims data. This clearly hinders research in this area of actuarial science. Second, individual claims data is censored, low-frequency and of a complex time-series structure. It is generally difficult to build good predictive models for such problems. 
Third, the claims reserving problem is a multi-period forecasting problem. However, often, the underlying algorithms are only trained for performing one-period ahead forecasts. Naturally, a tweak is required to work around this problem going from one- to multi-period forecasts.
Fourth, the implementation and structure of the proposed individual claims reserving methods is rather complex and often specific to a certain claims reserving situation, for instance, every insurance company collects historic data of a slightly different nature (and format). This makes it difficult to benchmark the different methods. Moreover, the proposed approaches often need extended hyper-parameter tuning, e.g., to avoid biases, this leaves the question open whether the proposed method easily generalizes to other claims reserving situations (in a broader sense).

This paper introduces a fundamentally novel approach which we envision as a transformative step toward the widespread adoption of individual claims reserving across the insurance industry. This transformative step is not about a specific ML architecture, but our core idea is to reorganize historical individual claims data for direct multi-period forecasting. The main step is to reformulate the foundational chain-ladder (CL) reserving algorithm so that one can perform multi-period model fitting and forecasting. Once this step is fully understood, adapting this idea to ML methods is straightforward.
We will explain this in detail, after outlining the present state of the field of individual claims reserving using ML methods.

\medskip

We observe four main techniques to cope with the multi-period forecasting problem in individual claims reserving:

\medskip

(1) The multi-period forecasting is performed by a recursive one-period forecast procedure using past observations as inputs. Rolling this recursive procedure into the future, missing observed inputs are replaced by their forecasts. This is the first and most popular method used for multi-period forecasting; for literature in individual claims reserving see, e.g.,  De Felice--Moriconi \cite{DeFeliceMoriconi} and Chaoubi et al.~\cite{Chaoubi}. We briefly explain why this procedure may be problematic. Consider observations $(Y_1,\ldots, Y_t)$ to forecast a next response $Y_{t+1}$ at time $t\ge 1$. The natural forecast is given by
\begin{equation*}
(Y_1,\ldots, Y_t) ~\mapsto~  \widehat{Y}_{t+1}:= \E\left[\left.Y_{t+1}\right| Y_1,\ldots, Y_t\right].
\end{equation*}
One period later, at time $t+1$, we have collected 
observations $(Y_1,\ldots, Y_{t+1})$ and we build the next forecast of response $Y_{t+1}$
by
\begin{equation}\label{forecast 1}
(Y_1,\ldots, Y_{t+1}) ~\mapsto~  \widehat{Y}_{t+2}= \E\left[\left.Y_{t+2}\right| Y_1,\ldots, Y_{t+1}\right].
\end{equation}
This gives a natural recursive one-period ahead forecast algorithm. The main difficulty in applying it in practice is that the observation $Y_{t+1}$ is not available at time $t$, and we cannot roll this algorithm into the future. The simple solution to this problem is to  impute the prediction $\widehat{Y}_{t+1}$ for the missing observation $Y_{t+1}$, that is, we set
\begin{equation}\label{forecast 1 forecast}
(Y_1,\ldots, Y_t, \widehat{Y}_{t+1}) ~\mapsto~  \widehat{Y}^\dagger_{t+2}:= \E\left[Y_{t+2}\left| Y_1,\ldots, Y_t, \widehat{Y}_{t+1}\right]\right..
\end{equation}
However, in general, this proposal of multi-period forecasting is inappropriate. We give an example. Assume that all responses are binary, $Y_t \in \{0,1\}$, $t\ge 1$. Thus, the conditional expectation in  \eqref{forecast 1} is based on binary observations
$(Y_1,\ldots, Y_t, Y_{t+1})  \in \{0,1\}^{t+1}$, and so is the ML model that is trained to approximate the forecast \eqref{forecast 1}. However, the forecast
$\widehat{Y}_{t+1} \in [0,1]$ imputed in \eqref{forecast 1 forecast} can take {\it any value} in the unit interval, e.g., 
$\widehat{Y}_{t+1} =0.46$, and the forecast model \eqref{forecast 1} does not know how to deal with this input value, because it has never seen a value different from zero or one before (because it only learned to deal with binary inputs).

\medskip

(2) A workaround of the problem discussed in item (1) is to learn a full simulation model from which one can simulate $Y_{t+1}$, given $Y_1,\ldots, Y_t$. This then allows one to perform a Monte Carlo simulation extrapolation. This is the solution applied, e.g., in 
W\"uthrich \cite{WTree} and Delong et al.~\cite{DelongLindholmW}. The main disadvantage of this approach clearly is that we need an accurate simulation model. If the responses $Y_t$ contain, e.g., claims payments, claims incurred and other stochastic processes, this is clearly beyond our modeling capabilities.

\medskip

(3) The works of Kuo \cite{Kuo1, Kuo2} present sequence-to-sequence forecasting methods, and the approach presented by Gabrielli \cite{GabrielliEAJ} uses a rather similar technique. These approaches mask missing observations, and the predictive model learns to perform forecasting under incomplete information, e.g., it tries to directly predict $Y_{t+2}$, given $(Y_1,\ldots, Y_t)$, by learning from all {\it available} information. This is a very suitable proposal. In practical applications, the main difficulty of this approach lies in controlling and mitigating potential biases during model training. Our proposal possesses this problem too, but we will see that expert intervention is rather easy in our approach.

\medskip

(4) Finally, an other option is to directly predict the ultimate claims from the available information. There are two approaches in the literature that consider this option. The first one is related to survival analysis that properly accounts for censored information. This has been considered, e.g.,  in Lopez et al.~\cite{Lopez, Lopez2}, Bladt--Pittarello \cite{Bladt}, Hiabu et al.~\cite{Hiabu} or Turcotte--Shi \cite{Turcotte}. The second method of a direct ultimate claim prediction uses reinforcement learning to optimally update the forecast based on the incoming information; see Avanzi et al.~\cite{AvanziRL}

\medskip

Our proposal aligns with option (4) of directly forecasting the ultimate claims, but the construction of our ultimate claim predictor is rather different from the two proposals above. Our starting point is the classic CL method. 
Usually, the CL method is rather accurate on aggregate claims -- verified by its successful use over many decades -- and the CL method is very simple in its use, not very prone to biases, easy to handle and easy to manipulate by expert knowledge. Our main contribution is to restructure the estimation of the CL factors. This reshapes the estimation procedure such that it can naturally be extended to ML methods on individual claims. Thus, our contribution is this novel representation of the CL estimation and not the ML method itself. In fact, to keep things simple we select an elementary neural network architecture in our example. We believe that our proposal widely opens the door for a natural pathway of ML applications to individual claims reserving. A few immediate advantageous over the existing methods are the following: (i) Our proposal can deal with any sort of stochastic dynamic covariates such as claims incurred or multiple payment processes of several injured people. (ii) It is a natural next step beyond the CL algorithm. This means that the basis is the CL method and our approach can easily be regularized using the CL predictor to control for biases. (iii) It can easily be extended to cash-flow forecasting in the sense of sequence-to-sequence methods as stated in item (3) above. In summary, the main advantage of starting from a CL structure is that the CL results give the guardrails for the ML predictions.

\medskip

A limitation that we should mention is that our proposal only considers
reported but not settled (RBNS) claims. This is common to most research in individual claims reserving, i.e., claims need to be reported so that their individual claims history is available for the use of the prediction of their further development. Incurred but not reported (IBNR) claims need to be forecast in addition, e.g., by a suitable frequency-severity reserving model.

\bigskip

{\bf Organization of the manuscript.}
In Section \ref{sec: Chain-ladder method}, we revisit Mack's CL model \cite{Mack}, and we present our main technical results that gives an alternative version of estimating the CL reserves. Section \ref{Individual ultimate prediction using machine learning} takes this alternative version to present a natural extension to individual claims reserving.
In Section \ref{Real data examples}, we study two small-scale real data examples that serve as proof of concept of our proposal. Finally, in Section
\ref{Conclusions and Outlook} we conclude and we give a list of next steps to lift this proposal to its full power for individual claims reserving. The mathematical proofs are given in the appendix.

\section{Chain-ladder method}
\label{sec: Chain-ladder method}
\subsection{Mack's distribution-free chain-ladder model}
We start by revisiting Mack's distribution-free CL model \cite{Mack}. Denote cumulative payments for claims in accident period $i \in \{1,\ldots, I\}$ and with development delay $j \in \{0,\ldots, J\}$ by $C_{i,j}$, and assume that these cumulative payments are strictly positive for all pairs $(i,j)$. Assume $I>J$, i.e., at least one accident period is fully observed (developed), and the goal is to predict the ultimate claims $C_{i,J}$ of accident periods $i>I-J$ at time $I$.

\begin{model ass}[CL model]
\label{distribution-free CL model assumptions}
The cumulative payment processes $(C_{i,j})_{0\le j\le J}$ of different accident periods $i \in \{1,\ldots, I\}$ are independent.
There exist positive parameters $(f_j)_{j=0}^{J-1}$ and  $(\sigma^2_j)_{j=0}^{J-1}$
such that for all $i \in \{1,\ldots, I\}$ and $j \in \{0, \ldots, J-1\}$ the following holds
\begin{eqnarray*}
\E \left[\left.C_{i,j+1}\right| C_{i,0}, \ldots, C_{i,j}
\right]&=& f_j\, C_{i,j},\\
{\rm Var} \left(\left.C_{i,j+1}\right| C_{i,0}, \ldots, C_{i,j}
\right)&=& \sigma^2_j\, C_{i,j}.
\end{eqnarray*}
\end{model ass}
An easy consequence of these assumptions is that one can compute the conditionally expected ultimate claims $C_{i,J}$ of accident periods $i>I-J$ at time $I$ as follows
\begin{equation}\label{true conditional mean}
\E \left[\left. C_{i,J} \right| C_{i,0}, \ldots, C_{i,I-i}\right]=
C_{i,I-i}\prod_{l=I-i}^{J-1} f_l,
\end{equation}
where $C_{i,0}, \ldots, C_{i,I-i}$ are the cumulative payments of accident period $i$ that are observed at time $I$, i.e., the cumulative payments $C_{i,j}$ with indices $i+j \le I$. This reflects the observed upper triangle\footnote{The shape of the observed cumulative payments $(C_{i,j})_{i+j\le I}$, is a triangle for $I=J+1$ and it is a trapezoid for $I>J+1$. To simplify language we generally speak about triangles, even if the shape is a trapezoid. This is common practice in claims reserving.} at time $I$.

To apply the CL method for reserving, there remains the estimation of the (unknown) {\it CL factors} $(f_j)_{j=0}^{J-1}$. At time $I$, these CL factors are commonly estimated by
\begin{equation}\label{CL factor estimates}
\widehat{f}_j = \frac{\sum_{i=1}^{I-j-1}C_{i,j+1}}
{\sum_{i=1}^{I-j-1}C_{i,j}}.
\end{equation}
This motivates the {\it CL predictors} for accident periods $i>I-J$ at time $I$ 
\begin{equation}\label{forward path}
\widehat{C}_{i,J}=
C_{i,I-i}\prod_{l=I-i}^{J-1} \widehat{f}_l.
\end{equation}
This is the classic CL predictor. It is unbiased for the (conditionally) expected ultimate claim; see Mack \cite{Mack} and W\"uthrich--Merz \cite[Lemma 3.3]{WM2008}.
We next give a different representation of the CL predictor  $\widehat{C}_{i,J}$, $i>I-J$, and this will pave the way to individual claims reserving utilizing ML techniques.

\subsection{Recursive chain-ladder factor estimation}
\label{Recursive chain-ladder factor estimation}
We can interpret the prediction in \eqref{forward path} as a forward-path estimation of the conditionally expected ultimate claim \eqref{true conditional mean}. This can be seen as follows
\begin{equation*}
\widehat{C}_{i,J}~=~C_{i,I-i}
\prod_{l=I-i}^{J-1} \widehat{f}_l~=~
\underbrace{\underbrace{C_{i,I-i}  \cdot \widehat{f}_{I-i}}_{I-i ~\to~ I-i+1}\,\cdot \,\widehat{f}_{I-i+1}}_{I-i ~\to~ I-i+2} \,\cdot\, \ldots \,\cdot \,\widehat{f}_{J-1}.
\end{equation*}
Thus, this is a step-wise (one-period ahead) roll forward extrapolation of $C_{i,I-i}$ using the CL factor estimates $(\widehat{f}_l)_{l=I-i}^{J-1}$ -- this is precisely the meaning of  {\it chain-ladder} extrapolation, see Figure \ref{fig:CL1}.

\begin{figure}[htb!]
\begin{center}
\includegraphics[width=\textwidth]{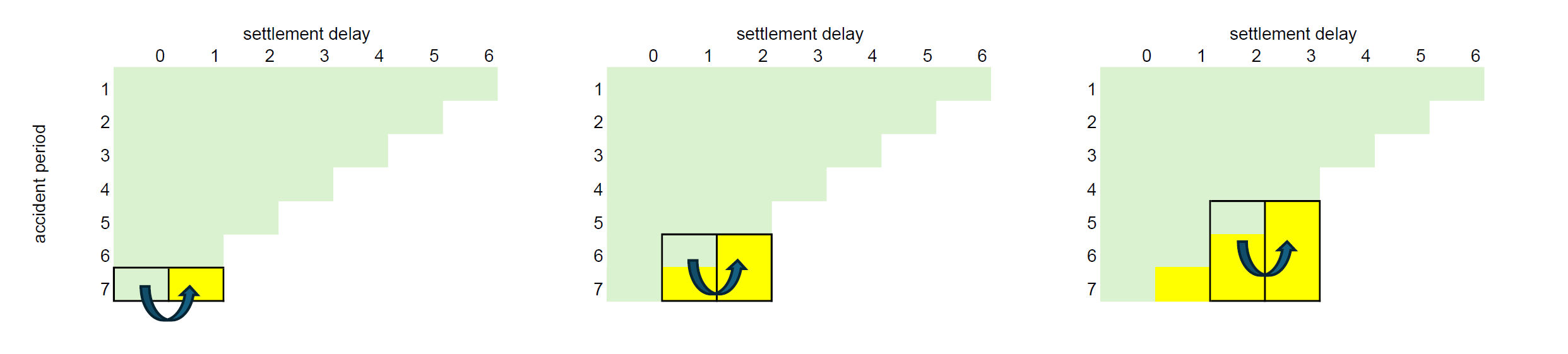}
\end{center}
\caption{Step-wise roll forward (chain-ladder) extrapolation to predict the ultimate claims $C_{i,J}$ using observations $C_{i,I-i}$, $i>I-J$, at time $I$ (for $I=7$ and $J=6$).}
\label{fig:CL1}
\end{figure}

We now give a different, recursive version that provides us with the same CL predictor by performing a {\it one-shot ultimate prediction}. Define the {\it projection-to-ultimate (PtU) factors} as follows
\begin{equation*}
  F_j := \prod_{l=j}^{J-1} f_l \qquad \text{ for $j\in \{0,\ldots, J-1\}$.}
\end{equation*}
Naturally, the following identity holds
\begin{equation*}
\E \left[\left. C_{i,J} \right| C_{i,0}, \ldots, C_{i,I-i}\right]~=~
C_{i,I-i}\prod_{l=I-i}^{J-1} f_l ~=~ 
C_{i,I-i} \, F_{I-i}.
\end{equation*}
The novel approach that we propose is to directly estimate these PtU factors $(F_j)_{j=0}^{J-1}$.
For this, we begin by introducing a new notation for the ultimate claims that are fully observed at time $I$, this will simplify the notation in the recursion below. Set
\begin{equation*}
\widehat{C}^*_{i,J}=C_{i,J}  \qquad \text{ for all $i \in \{1,\ldots, I-J\}$.}
\end{equation*}
We emphasize that these variables are all observed at time $I$; they correspond to the fully developed accident periods at time $I$.

\medskip

{\it (a) Initialization.}
We perform the estimation procedure of the PtU factors $(F_j)_{j=0}^{J-1}$ recursively starting from the upper-right corner of the claims reserving triangle.
We estimate the PtU factor of the last development period $j=J-1$ by
\begin{equation}\label{CL2 A}
\widehat{F}_{J-1}=\widehat{f}_{J-1}= \frac{\sum_{i=1}^{I-J}\widehat{C}^*_{i,J}}
{\sum_{i=1}^{I-J}C_{i,J-1}}
\qquad \text{ and } \qquad \widehat{C}^*_{I-J+1,J}=C_{I-J+1,J-1}\,\widehat{F}_{J-1}.
\end{equation}
{\it (b) Iteration.} The estimation is extrapolated recursively from $j+1 \to j$ by setting
\begin{equation}\label{CL2 B}
\widehat{F}_{j}= \frac{\sum_{i=1}^{I-j-1}\widehat{C}^*_{i,J}}
{\sum_{i=1}^{I-j-1}C_{i,j}}
\qquad \text{ and } \qquad \widehat{C}^*_{I-j,J}=C_{I-j,j}\,\widehat{F}_{j},
\end{equation}
for $j\in \{0,\ldots, J-2\}$.

\begin{figure}[htb!]
\begin{center}
\includegraphics[width=\textwidth]{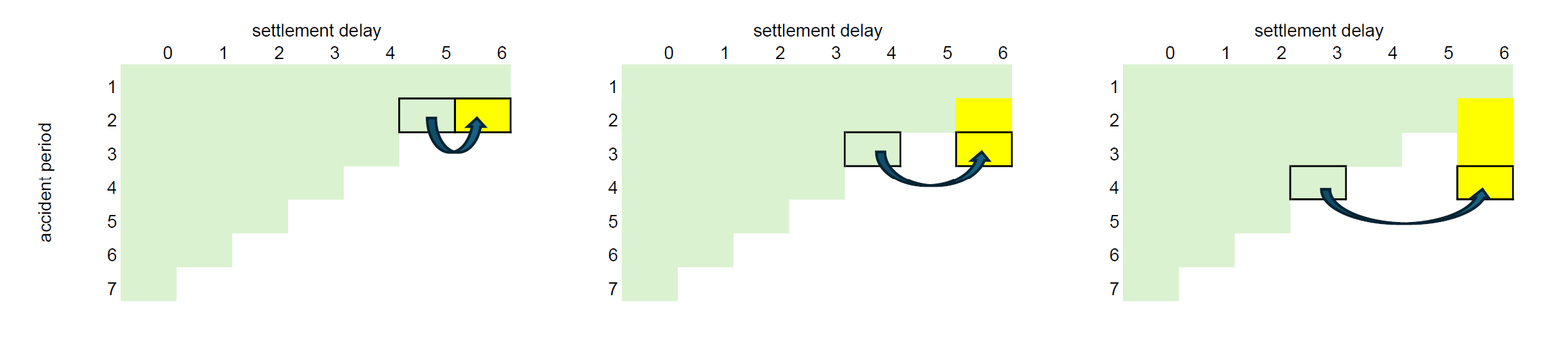}
\end{center}
\caption{Backward extrapolation to predict the ultimate claims $C_{i,J}$, $i>I-J$, using the estimated PtU factor estimates $(\widehat{F}_j)_{j=0}^{J-1}$: (left-middle-right) correspond to $j=J-1=5$, $j=4$ and $j=3$.}
\label{fig:CL2}
\end{figure}

The estimation of the forecast process \eqref{CL2 A}-\eqref{CL2 B} starts in the upper-right corner of the claims reserving triangle $j=J-1$, see Figure \ref{fig:CL2}. Then it moves recursively $j+1 \to j$ along the last observed diagonal to the most recent accident period $i=I$ ($j=0$) by {\it directly} predicting the ultimate claims using \eqref{CL2 B}.

\begin{prop}\label{CL prop}
$\widehat{C}^*_{i,J}=\widehat{C}_{i,J}$ for all accident periods $i\in \{I-J+1, \ldots, I\}$.
\end{prop}
The proof of this proposition is given in the appendix, and this is the result proved in Lorenz--Schmidt \cite{LorenzSchmidt}.

\medskip

The exciting fact stated in Proposition \ref{CL prop} is that both algorithms -- the step-wise roll forward CL extrapolation of Figure \ref{fig:CL1} and the backward extrapolation of Figure \ref{fig:CL2} -- lead to the identical CL reserves. We give some remarks.
\begin{itemize}
\item Both algorithms shown in Figures \ref{fig:CL1} and \ref{fig:CL2} use the identical information and they provide the identical forecasts. We can view this as an optimal use of triangular information because every observed cumulative payment $C_{i,j}$, $i+j \le I$, is used whenever it is appropriate. This is different, e.g., from the bootstrap of individual claims histories of Rosenlund \cite{Rosenlund} who bases his model learning procedure only on claims with fully observed trajectories.
\item The claims reserving problem is a typical situation of censored observations, with the additional difficulty that different (censored) observations have different maturity. This gives the close connection to survival analysis. We mention the Kaplan--Meier weighting method considered in Lopez et al.~\cite{Lopez, Lopez2} to correct for censored claims information, although the method in these references is not directly applicable to claims reserving because cumulative claims payments are assumed to be proportional to time to settlement. Other survival analysis papers of a similar nature are Hiabu et al.~\cite{Hiabu} -- this paper considers survival analysis for IBNR prediction using a CL forward path -- and Bladt--Pittarello \cite{Bladt} use the Aalen--Johansen estimator in a multi-state process on a finite state space.
\end{itemize}  

\subsection{Chain-ladder RBNS reserving}
\label{Chain-ladder RBNS reserving}
The ultimate claim predictions $\widehat{C}_{i,J}$ of the CL algorithm cover both the RBNS and the IBNR claims.
Our first modification is to only predict ultimates of RBNS claims at time $I$; by RBNS claim we denote any reported claim at time $I$ which can either be a closed or an open one (i.e., we allow for re-openings in this definition). 
We label individual claims by $\nu$, separately for each occurrence period $i \in \{1,\ldots, I\}$.
Let $C_{i,j|\nu}$ denote the cumulative payments of the $\nu$-th claim of occurrence period $i$ after development period $j$ (since occurrence). We can then write for the aggregate cumulative payments 
\begin{equation*}
  C_{i,j} = \sum_{\nu:\, T_{i|\nu} \le j} C_{i,j|\nu},
  \end{equation*}
  where the sum runs over all claims $\nu$ of accident period $i$ that are reported at time $i+j$, and
  $T_{i|\nu} \in \{0,\ldots, J\}$ denotes the reporting delay of the $\nu$-th claim of accident period $i$.

The reported claims at time $I$ are characterized by $i+T_{i|\nu} \le I$ and the IBNR claims by
  $i+T_{i|\nu} > I$. We now focus on RBNS reserves that only involve reported claims. For this we define a version of the PtU factor estimates $(\widehat{F}^{(r)}_j)_{j=0}^{J-1}$ that {\it only} consider the reported claims in the corresponding periods. 
  
\medskip  
  
{\it (a) Initialization.}  We set $\widehat{C}^{(r)}_{i,J|\nu}=C_{i,J|\nu}$ for all claims $\nu$ with accident dates $i\le I-J$; these claims are fully developed. 
We then initialize for $j=J-1$
\begin{equation}\label{CL2 A reported}
\widehat{F}^{(r)}_{J-1}= \frac{\sum_{i=1}^{I-J} \sum_{\nu:\, T_{i|\nu} \le J-1}\widehat{C}^{(r)}_{i,J|\nu}}
{\sum_{i=1}^{I-J}\sum_{\nu:\, T_{i|\nu} \le J-1}C_{i,J-1|\nu}}
\qquad \text{ and } \qquad \widehat{C}^{(r)}_{I-J+1,J|\nu}=C_{I-J+1,J-1|\nu}\,\widehat{F}^{(r)}_{J-1},
\end{equation}
for all reported claims $\nu$ of accident period $I-J+1$, i.e., with reporting delay $T_{I-J+1|\nu} \le J-1$. The PtU factor estimate $\widehat{F}^{(r)}_{J-1}$
only considers claims that have a reporting delay less or equal to $J-1$, thus, in the numerator and the nominator of the PtU factor estimate $\widehat{F}^{(r)}_{J-1}$ we consider the {\it identical cohort} of reported claims.
Thus, this is a consistent consideration of reported claims extrapolation.

\medskip

{\it (b) Iteration.}
This is recursively extended from $j+1 \to j \in\{0,\ldots, J-2\}$ by setting
\begin{equation}\label{CL2 B reported}
\widehat{F}^{(r)}_{j}= \frac{\sum_{i=1}^{I-j-1}\sum_{\nu:\, T_{i|\nu} \le j}\widehat{C}^{(r)}_{i,J|\nu}}
{\sum_{i=1}^{I-j-1}\sum_{\nu:\, T_{i|\nu} \le j}C_{i,j|\nu}}
\qquad \text{ and } \qquad \widehat{C}^{(r)}_{I-j,J|\nu}=C_{I-j,j|\nu}\,\widehat{F}^{(r)}_{j},
\end{equation}
for all reported claims $\nu$ of accident period $I-j$, i.e., with reporting delay $T_{I-j|\nu} \le j$.
Again, in the numerator and the nominator of the PtU factor estimate $\widehat{F}^{(r)}_{J-1}$ we consider the {\it identical cohort} of reported claims $\nu$, i.e., with $T_{i|\nu} \le j$.
We can also interpret this of having a moving
target
\begin{equation*}
\sum_{\nu:\, T_{i|\nu} \le j}\widehat{C}^{(r)}_{i,J|\nu},
\end{equation*}
in the numerators of \eqref{CL2 B reported} as this claims cohort is decreasing for decreasing $j$ due increasingly many IBNR claims $T_{i|\nu} > j$ for decreasing $j$.

\medskip

We would like to highlight the similarity of our proposal to Schnieper's model \cite{Schnieper}, although the estimation procedures differ. Schnieper's model separates RBNS from IBNR claims similarly to our PtU estimation, but Schnieper's estimation of the development factors contains reserves for IBNR claims (once they are reported), whereas our estimation procedure \eqref{CL2 B reported} fixes the claims cohort and projects it directly to its ultimates $\widehat{C}^{(r)}_{i,J|\nu}$ without including reserves for additional late reported claims.

\section{Individual ultimate prediction using machine learning}
\label{Individual ultimate prediction using machine learning}
The recursive formulas \eqref{CL2 A reported}-\eqref{CL2 B reported} give the intuition for how the CL method can be extended to an individual claims ML method using {\it any} available input information. Denote the settlement process of claim $\nu$ of accident period $i$ by ${\cal C}_{i|\nu}:=(C_{i,l|\nu}, \bX_{i,l|\nu})_{l=0}^{J}$,
where $(\bX_{i,l|\nu})_{l=0}^{J}$ can be any claims feature process that may involve {\it static covariates} that are known at reporting (like claim type or business line), {\it deterministic dynamic covariates} (like settlement delay) and {\it stochastic dynamic covariates} (like claims incurred process or claim status process for open/closed claims).
We emphasize that this latter information typically is not fully known  at time $I$ for RBNS claims with accident dates $i>I-J$ as it stochastically evolves into the future, and one may want to consider its stochastic modeling and prediction as well. The exciting fact about our method presented below is that extrapolation of the claims feature process is {\bf not} necessary!

\subsection{Recursive individual claims reserving}

{\bf Assumption.} We assume that the individual claims processes ${\cal C}_{i|\nu}$ are independent, and that they are conditionally i.i.d., given the static covariates known at reporting. 

\medskip

To comply with the previous assumption it may require to adjust some of the dynamic quantities with inflation indices (e.g., payments, incurred and case reserves). For late reported RBNS claims, we mask the covariate process part of $(\bX_{i,l|\nu})_{l=0}^{J}$ that refers to time periods $j < T_{i|\nu}$ before reporting.

\bigskip

We now present a recursive forecast procedure for individual claims that in its methodology is equivalent to \eqref{CL2 A reported}-\eqref{CL2 B reported}.

\medskip

{\it (a) Initialization.} We start with accident period $I-J+1$ or, equivalently, development period $j=J-1$. For all claims $\nu$ that have occurred in accident periods $i \le I-J$, the stochastic processes ${\cal C}_{i|\nu}=(C_{i,l|\nu}, \bX_{i,l|\nu})_{l=0}^{J}$ are fully observed, and we aim at predicting the ultimates $C_{I-(J-1),J|\nu}$ of all claims in the next accident period $I-J+1$. These claims are reported at time $I$, thus, they have reporting delay $T_{I-(J-1)|\nu} \le J-1$, and we have observed their claims histories $(C_{I-(J-1),l|\nu}, \bX_{I-(J-1),l|\nu})_{l=0}^{J-1}$.

To bring the learning data into the same shape, we are only allowed to consider those claims that have been reported before settlement delay $J$. This gives us the learning data to perform the initial learning step
\begin{equation*}
{\cal L}_{J-1} = \left\{ \left(C_{i,J|\nu}, (C_{i,l|\nu}, \bX_{i,l|\nu})_{l=0}^{J-1}\right);\, i\le I-J \text{ and } T_{i|\nu} \le J-1 \right\}.
\end{equation*}
Based on this learning data ${\cal L}_{J-1}$, we build the regression function
\begin{equation}\label{regression J-1}
(C_{i,l|\nu}, \bX_{i,l|\nu})_{l=0}^{J-1}
~\mapsto~ \mu_{J-1}\left((C_{i,l|\nu}, \bX_{i,l|\nu})_{l=0}^{J-1}\right)=
\E \left[C_{i,J|\nu} \left|(C_{i,l|\nu}, \bX_{i,l|\nu})_{l=0}^{J-1}
\right]\right. .
\end{equation}
This regression function $\mu_{J-1}(\cdot)$ describes the forecast step $J-1 \to J$, and it allows us to predict the ultimates for the claims $\nu$ of accident period $i=I-J+1$, with  observed claims histories
$(C_{I-J+1,l|\nu}, \bX_{I-J+1,l|\nu})_{l=0}^{J-1}$. In particular, we can compute the forecast (assuming stationary along the occurrence period axis $i$)
\begin{eqnarray*}
\widehat{C}_{I-(J-1),J|\nu}
&:=&\mu_{J-1}\left((C_{I-(J-1),l|\nu}, \bX_{I-(J-1),l|\nu})_{l=0}^{J-1}\right)
\\&=&
\E \left[C_{I-(J-1),J|\nu} \left|(C_{I-(J-1),l|\nu}, \bX_{I-(J-1),l|\nu})_{l=0}^{J-1}
\right]\right. .
\end{eqnarray*}
We append these predictions to the observed claims histories for accident period $I-J+1$
\begin{equation}\label{append 1}
  \widehat{\cal C}_{I-(J-1)|\nu}:=
\left((C_{I-(J-1),l|\nu}, \bX_{I-(J-1),l|\nu})_{l=0}^{J-1}, \,\widehat{C}_{I-(J-1),J|\nu}\right).
\end{equation}
This initializes the recursion, and it reflects the figure on the left-hand side
of Figure \ref{fig:CL2}.

\medskip


{\it (b) Iteration.}  We recursively extend this from $j+1 \to j \in \{0,\ldots, J-2\}$ or, equivalently, to accident period $I-j$. The stochastic processes ${\cal C}_{i|\nu}$ are fully observed for accident periods $i \le I-J$, and for accident periods $I-(J-1) \le i \le I-(j+1)$, we have partially observed claims histories $\widehat{\cal C}_{i|\nu}$ including an appended ultimate claim forecast $\widehat{C}_{i,J|\nu}$. The goal is to extend this to the claims $\nu$ of accident period $I-j$, which have a reporting delay $T_{I-j|\nu} \le j$ and observed claims histories $(C_{I-j,l|\nu}, \bX_{I-j,l|\nu})_{l=0}^{j}$ at time $I$.

We select the learning dataset such that it aligns with the maximal reporting lag of the RBNS claims of accident period $I-j$ at time $I$ to build consistent claims cohorts
\begin{eqnarray}\nonumber
{\cal L}_{j} &=& \left\{ \left(C_{i,J|\nu}, (C_{i,l|\nu}, \bX_{i,l|\nu})_{l=0}^{j}\right);\, i\le I-J \text{ and } T_{i|\nu} \le j \right\}
\\&&\cup \,
\left\{ \left(\widehat{C}_{i,J|\nu}, (C_{i,l|\nu}, \bX_{i,l|\nu})_{l=0}^{j}\right);\, I-(J-1)\le i\le I-(j+1) \text{ and } T_{i|\nu} \le j \right\}.\qquad
\label{appended history}
\end{eqnarray}
Based on this learning data ${\cal L}_{J-1}$, we build the regression function, which for the fully developed claims of accident periods $i \le I-J$ considers
\begin{equation}\label{estimate 1}
(C_{i,l|\nu}, \bX_{i,l|\nu})_{l=0}^{j}
~\mapsto~ \mu_{j}\left((C_{i,l|\nu}, \bX_{i,l|\nu})_{l=0}^{j}\right)=
\E \left[C_{i,J|\nu} \left|(C_{i,l|\nu}, \bX_{i,l|\nu})_{l=0}^{j}
\right]\right., 
\end{equation}
this predicts from development period $j \to J$. For the claims of accident periods
$I-(J-1)\le i\le I-(j+1)$ we use  the modified version
\begin{eqnarray}\nonumber
(C_{i,l|\nu}, \bX_{i,l|\nu})_{l=0}^{j}
~\mapsto~ \mu_{j}\left((C_{i,l|\nu}, \bX_{i,l|\nu})_{l=0}^{j}\right)
&=&\nonumber
\E \left[C_{i,J|\nu} \left|(C_{i,l|\nu}, \bX_{i,l|\nu})_{l=0}^{j}
\right]\right.
\\&=&
\E \left[\widehat{C}_{i,J|\nu} \left|(C_{i,l|\nu}, \bX_{i,l|\nu})_{l=0}^{j}
\right]\right.,
 \label{estimate 2}
\end{eqnarray}
the last step uses the tower property of conditional expectation. The latter can be learned from the second line of the learning sample ${\cal L}_{j}$ given in \eqref{appended history}.
Using this learned regression function $\mu_{j}(\cdot)$, given in \eqref{estimate 1}-\eqref{estimate 2}, we forecast the ultimate claims of accident period $I-j$ by
\begin{eqnarray*}
\widehat{C}_{I-j,J|\nu}
:=\mu_{j}\left((C_{I-j,l|\nu}, \bX_{I-j,l|\nu})_{l=0}^{j}\right)
=\E \left[C_{I-j,J|\nu} \left|(C_{I-j,l|\nu}, \bX_{I-j,l|\nu})_{l=0}^{j}
\right]\right. .
\end{eqnarray*}
This ultimate claim forecast is appended to the available claims histories
\begin{equation}\label{append 2}
  \widehat{\cal C}_{I-j|\nu}:=\left((C_{I-j,l|\nu}, \bX_{I-j,l|\nu})_{l=0}^{j},\, \widehat{C}_{I-j,J|\nu}\right).
\end{equation}

\medskip

Recursive iteration completes the ultimate claim forecasting, precisely reflecting the philosophy of Section \ref{Recursive chain-ladder factor estimation}. This is illustrated by the middle and the right-hand side figures of Figure \ref{fig:CL2}.

\begin{rems}\normalfont~\label{remarks 1}
\begin{itemize}
\item Crucially, the above algorithm can deal with {\it any} dynamic stochastic covariates {\it without} their extrapolation.
\item The above algorithm is not about a certain ML architecture, but it is rather about how the data is organized to perform the one-shot ultimate prediction. The input $(\bX_{i,l|\nu})_{l=0}^{j}$ can be of any (reasonable) form, it can include stochastic dynamic covariates, it can involve textual data like medical reports, but it can also include information at a higher time frequency, i.e., the forecasting frequency periods $j\ge 0$ can also deal with continuously arriving information $(\bX_{i,l|\nu})_{l \in [0,j]}$. At this stage the choice of the ML architecture will become important as it needs to be able to deal with the formats of the inputs.

\item In practical applications, the only critical item of this algorithm is its recursive nature. In \eqref{append 1} and \eqref{append 2} we append the estimated ultimate claim of a given accident period $i=I-(j+1)$ to the observations, in order to be able to learn the one of the next accident period $i+1=I-j$. This recursive nature has the disadvantage that if one accident period has a biased estimate, this bias will propagate through the following accident period estimates. Therefore, it is very important to correct for biases whenever possible.
\item This recursive algorithm uses the observed data in the same efficient way as the CL method, this is motivated by the result of Proposition \ref{CL prop}. In particular, we simultaneously learn from fully developed but also from non-settled claims, such that always all relevant information is considered for prediction. The trade-off of this optimal use of information is the bias-control mentioned in the previous item.
\end{itemize}
\end{rems}

\section{Real data examples}
\label{Real data examples}
We present two small scale real data examples. The first one considers accident insurance and the second one liability insurance. These two examples serve as a proof of concept, and they are neither meant to solve the most complex claims reserving problem, nor do they use the latest most fancy ML techniques, see second item of Remarks \ref{remarks 1}. As mentioned, our main contribution is the optimal use of triangular data for one-shot ultimate claim prediction which leads to our key result of Proposition \ref{CL prop}. This widely opens the door for individual claims reserving, and in our final Section \ref{Conclusions and Outlook}, below, we highlight potential next steps.

To further strengthen our proposal, we select two examples for which also the lower triangle is known. This allows us to benchmark our forecasts against the true outcomes. Naturally, this requires that we restrict ourselves to older accident periods as the most recent accident period considered needs to be fully developed to be able to benchmark it against the ground truth.

\subsection{Description of data}
\subsubsection{Accident insurance data}
We use an accident insurance dataset on an annual scale with 5 fully observed accident years, i.e., we have a fully observed $5\times 5$ square. For model fitting and forecasting, we {\it only} use the {\it upper triangle},
as in Figure \ref{fig:CL2}, and we benchmark the forecasts against the true ultimates which are available here (having also observed the lower triangle).

\begin{table}[h]
\footnotesize
\centering
\begin{tabular}{lc}
\toprule
\textbf{Characteristic} &  \\
\midrule
Time scale & calendar years  \\
Number of accident years & 5  \\
Number of development years & 5  \\
Number of reported claims & 66,639  \\
\midrule
\multicolumn{2}{l}{\textbf{Data description}} \\
\midrule
\multicolumn{2}{l}{Annual individual cumulative payments $C_{i,j|\nu}$} \\
\multicolumn{2}{l}{Claim status $O_{i,j|\nu}\in \{0,1\}$ for closed/open} \\
\multicolumn{2}{l}{Binary static covariate for work or leisure accident} \\
\multicolumn{2}{l}{Calendar month of accident} \\
\multicolumn{2}{l}{Reporting delay in daily units} \\
\bottomrule
\end{tabular}
\caption{Characteristics of accident dataset.}
\label{tab:accidentdata}
\end{table}
Table \ref{tab:accidentdata} shows the available data of the accident insurance dataset. There are 66,639 reported claims with a fully observed development history over the $5\times 5$ square. Besides the individual cumulative payment process $(C_{i,j|\nu})_j$, there is information about the claim status process $(O_{i,j|\nu})_j$, with $O_{i,j|\nu}=1$ meaning that the $\nu$-th claim of accident year $i$ is open at the end of settlement delay $j$, and closed otherwise. Then, there is static information about: work or leisure related accident, the calendar month of the accident and the reporting delay in daily units. We collect all this static and dynamic covariate information in the claims feature process $(\bX_{i,j|\nu})_j$, see Section \ref{Individual ultimate prediction using machine learning}.

We give some remarks on the available covariates. The claim status process $(O_{i,j|\nu})_j$ is stochastic dynamic. Typically, for closed claims $O_{i,j|\nu}=0$ we do not expect further payments beyond settlement delay $j$ and the cumulative payment process should remain constant in these further periods. However, claims can be re-opened for further (unexpected) claims developments. Our method also reserves for these (potential) payments.
 Typically, work or leisure related accidents are different, e.g., a bank employee cannot have a skiing accident during work hours. The calendar month will distinguish summer from winter activities (biking vs.~skiing accidents). Finally, the reporting delay is important to infer the waiting period of the insurance contract, e.g., if the contract has a waiting period of 3 months, then daily allowance is only paid after this waiting period. Typically, the waiting period is positively correlated with the reporting delay.

\subsubsection{Liability insurance data}
The second dataset considers liability insurance (excluding motor liability). We again have a fully observed $5\times 5$ square and for model fitting we only use the upper triangle.

\begin{table}[h]
\footnotesize
\centering
\begin{tabular}{lc}
\toprule
\textbf{Characteristic} &  \\
\midrule
Time scale & calendar years  \\
Number of accident years & 5  \\
Number of development years & 5  \\
Number of reported claims & 21,991  \\
\midrule
\multicolumn{2}{l}{\textbf{Data description}} \\
\midrule
\multicolumn{2}{l}{Annual individual cumulative payments $C_{i,j|\nu}$} \\
\multicolumn{2}{l}{Claim status $O_{i,j|\nu}\in \{0,1\}$ for closed/open} \\
\multicolumn{2}{l}{Claim incurred $I_{i,j|\nu}\ge 0$} \\
\multicolumn{2}{l}{Binary static covariate for private vs.~commercial liability} \\
\multicolumn{2}{l}{Calendar month of accident} \\
\multicolumn{2}{l}{Reporting delay in daily units} \\
\bottomrule
\end{tabular}
\caption{Characteristics of liability dataset.}
\label{tab:liablitydata}
\end{table}

Table \ref{tab:liablitydata} shows the available data of the liability insurance dataset. The main difference to the previous example is that for this dataset there is also a claims incurred process $(I_{i,j|\nu})_j$ available. The claims incurred process is a claims adjuster's prediction of the individual ultimate claim that is constantly updated when new information becomes available, i.e., this is a stochastic process driven by the claims adjuster's assessments. Additionally, there is the case reserve process 
$(R_{i,j|\nu})_j$ available by computing $R_{i,j|\nu}=I_{i,j|\nu}-C_{i,j|\nu}$.

\subsection{Forecast model}
\subsubsection{Network architecture}
Next, we select the regression functions $(\mu_j)_{j=0}^{J-1}$; we refer to \eqref{regression J-1} and \eqref{estimate 1}-\eqref{estimate 2}.
In our two small scale examples we have $I=5$ and $J=4$, thus, we only need 4 regression functions $(\mu_j)_{j=0}^3$ in each of the two examples. We give a crude proposal that clearly allows for modification and improvement for bigger and refined datasets, see also Remarks \ref{remarks 1} and the discussion in Section \ref{Conclusions and Outlook}. Namely, we model all 4 regression functions $(\mu_j)_{j=0}^3$ by separate networks. Hence, we have to learn 4 different networks for each of the two examples. Another simplification that we make is the following Markov assumption
\begin{equation}\label{Markov}
\mu_{j}\left((C_{i,l|\nu}, \bX_{i,l|\nu})_{l=0}^{j}\right)=
\mu_{j}\left(C_{i,j|\nu}, \bX_{i,j|\nu}\right).
\end{equation}
That is, we do not consider the entire individual claims history $(C_{i,l|\nu}, \bX_{i,l|\nu})_{l=0}^{j}$, but we assume that it is sufficient to know the last state $(C_{i,j|\nu}, \bX_{i,j|\nu})$.

A natural choice for $\mu_{j}$ under the Markov assumption \eqref{Markov} is a feed-forward neural network (FNN) architecture; we will use the FNN notation and terminology of W\"uthrich--Merz \cite{WM2023}, and for a broader introduction to networks we also refer to that reference. If the Markov assumption \eqref{Markov} fails to hold, then a natural candidate for $\mu_j$ is a Transformer architecture with an integrated CLS token for input encoding; see Richman et al.~\cite{RichmanCred}.

\begin{table}[h]
\centering
{\footnotesize
\begin{center}
\begin{tabular}{|l||c|c|c|}
\hline
Module & Dimension & $\#$ Weights & Activation
\\\hline\hline
Input layer & 5 & --& -- \\
1st hidden layer & 20 & 120 & $\tanh$ \\
2nd hidden layer & 15 & 315 & $\tanh$ \\
3rd hidden layer & 10 & 160 & $\tanh$ \\
Output layer & 1 & 11 & $\exp$ \\
\hline
\end{tabular}
\end{center}}
\caption{Selected FNN architecture $\mu^{\rm FNN}_{j}$, for $1\le j\le J-1$, in the accident insurance example.}
\label{architecture table}
\end{table}

We select identical plain-vanilla FNN architectures 
\begin{equation*}
(C_{i,j|\nu}, \bX_{i,j|\nu})~\mapsto~
\mu^{\rm FNN}_{j}(C_{i,j|\nu}, \bX_{i,j|\nu})
=\exp \left(\bz_j^{(3:1)}(C_{i,j|\nu}, \bX_{i,j|\nu})\right),
\end{equation*}
for all development periods $j=0,\ldots, 3$. These FNN architectures  $\mu^{\rm FNN}_{j}$ consider FNNs $\bz_j^{(3:1)}$ with 3 hidden layers with number of neurons given by $(20,15,10)$ in the 3 hidden layers. The explicit specifications are shown in Table \ref{architecture table}, and this results in FNNs with 606 weights to be fitted in the case of the accident insurance dataset for every $j=0,\ldots, 3$. Since under the Markov assumption \eqref{Markov}, the input has always the same dimension, we use the identical FNN architecture for all regression functions  $(\mu^{\rm FNN}_j)_{j=0}^3$, only their weights will differ (after model fitting). For an explicit implementation of this architecture we refer to W\"uthrich--Merz \cite[Listing 7.1]{WM2023}, only the part involving the volume in that reference needs to be dropped because here we only have unit exposures.

We select the same FNN architecture also for the liability insurance dataset. In this second dataset we have an input dimension of 7 (additionally considering claims incurred and case reserves). This results in 160 weights in the first hidden layer and a FNN architecture of totally having 646 weights.

\subsubsection{Data pre-processing and model fitting}
\label{Data pre-processing and model fitting}
We are almost ready now to fit these four FNNs $(\mu^{\rm FNN}_j)_{j=0}^3$ for both examples. There remains the input pre-processing and the selection of the fitting procedure. Typically, inputs $(C_{i,j|\nu}, \bX_{i,j|\nu})$ should be standardized to receive an efficient stochastic gradient descent (SGD) fitting procedure. The individual cumulative payments are considered on the log-scale and they are standardized by
\begin{equation*}
\widetilde{C}_{i,j|\nu} ~\leftarrow~
\frac{\log\left( \max\{1,C_{i,j|\nu}\}\right) - \widehat{m}}{\widehat{s}}.
\end{equation*}
where $\widehat{m}$ and $\widehat{s}$ are the empirical mean and standard deviation of $\log ( \max\{1,C_{i,j|\nu}\})$ over all {\it available} cumulative payments $C_{i,j|\nu}$ simultaneously in all accident years $i$, for all claims $\nu$ and in all development periods $j$. 

In the case of the liability dataset, we apply the same transformation to claims incurred $I_{i,j|\nu}$, and we apply the standardization to the case reserves $R_{i,j|\nu}$ without going to the log-scale.

Furthermore, the claim status $O_{i,j|\nu} \in \{0,1\}$ does not require pre-processing. The binary static covariates are mapped to $\{0,1\}$, the calendar month $cm \in \{1,\ldots, 12\}$ is scaled to $[0,1]$ by transforming it to $(cm-1)/11$, and the reporting delay is censored at 365 days, mapped to the log-scale and being scaled to $[0,1]$ as for the calender month.
This provides us with the pre-processed input $(\widetilde{C}_{i,j|\nu}, \widetilde{\bX}_{i,j|\nu})$ which is used to forecast the corresponding ultimate claim
\begin{equation*}
(\widetilde{C}_{i,j|\nu}, \widetilde{\bX}_{i,j|\nu}) ~\mapsto ~
\widehat{C}^{\rm FNN}_{i,J|\nu}=
\mu_j^{\rm FNN}\left(\widetilde{C}_{i,j|\nu}, \widetilde{\bX}_{i,j|\nu}\right).
\end{equation*}
These FNN regression functions are learned recursively as described in 
Section \ref{Individual ultimate prediction using machine learning}, see also Figure \ref{fig:CL2}. For fitting, we apply the SGD algorithm with early stopping using the hyper-parameter specification as reported in Table \ref{hyperparameters}.

\begin{table}[h]
\footnotesize
\centering
\begin{tabular}{ll}
\toprule
\textbf{Component} & \textbf{Setting} \\
\midrule
Loss function & mean squared error (MSE)\\
Optimizer & Adam with learning rate $10^{-3}$ \\
Batch size and epochs & 4,096 and 1,000\\
Learning-validation split & $9:1$\\
Early stopping & reduce learning rate on plateau, factor 0.9, patience 5\\
Ensembling & 10 network fits with different seeds \\
\bottomrule
\end{tabular}
\caption{Key implementation hyper-parameters for FNN fitting.}
\label{hyperparameters}
\end{table}

\subsubsection{Bias control}
\label{Bias control}
In the forecast procedure, we install one special feature that is very crucial in bias control. Assuming stationarity along the accident year axis for given static features, there should not be any trend in the ultimate claim predictions, of course, this also means that inflation-adjusted quantities are considered. 

It is a well-known fact that SGD fitting with early stopping leads to biased models, see, e.g., W\"uthrich \cite{Wbalance}. We fix this by ensuring that the balance property holds, by shifting the prediction by the size of the observed in-sample bias. That is, 
we fit the FNN architecture on the learning data ${\cal L}_j$ using an early stopped SGD providing us with the fitted FNN architecture $\widehat{\mu}_j^{\rm FNN}(\cdot)$.
A balance corrected predictive version thereof is obtained by setting
\begin{equation}\label{bc predictor}
\widehat{\mu}_j^{{\rm bc-FNN}}(\cdot) = \frac{\sum_{i \le I-J} \sum_\nu C_{i,J|\nu}+
\sum_{I-(J-1) \le i \le I-(j+1)} \sum_\nu \widehat{C}_{i,J|\nu}}{\sum_{i \le I-(j+1)} \sum_\nu\widehat{\mu}_j^{\rm FNN}
\left((C_{i,l|\nu}, \bX_{i,l|\nu})_{l=0}^{j}\right)}\,
\widehat{\mu}_j^{\rm FNN}(\cdot).
\end{equation}
This gives a multiplicative correction to align the average in-sample prediction with the average observed (estimated) response; note that this is again of a recursive nature.

\medskip

We give two further features that may help to improve the predictive models.
\begin{itemize}
\item Because network fitting involves many elements of randomness, e.g., the initialization of the SGD algorithm, we always ensemble over 10 balance corrected predictors \eqref{bc predictor} being received from the same SGD algorithm but with  different seeds for initialization; see Richman--W\"uthrich \cite{RichmanW2}.
\item The fitting procedure can be regularized using expert knowledge, e.g., if we have a strong prediction from the CL or the Bornhuetter--Ferguson \cite{BF} method, we can use this prediction to regularize the learned FNN predictor 
$\widehat{\mu}_j^{\rm FNN}(\cdot)$. This can be done on an individual claims level, but also for more coarse claims cohorts.
\end{itemize}

\subsection{Results}
\subsubsection{Chain-ladder results}
We start by reporting the classic CL results of Mack \cite{Mack} on cumulative payments. This will set the benchmark for all subsequent methods.

\begin{table}[h]
\centering
{\footnotesize
\begin{center}
\begin{tabular}{|l|r|rr|rr|}
\hline
 & True OLL & CL reserves  & RMSEP & Error & in \% \\
\hline\hline
\underline{Accident} &&&&&\\
Mack's CL model \cite{Mack} & 24,212&	23,064&1,663&-1,148 &69\%\\
RBNS CL method &19,733 &18,959	&&	-774 &\\	
 \hline\hline
\underline{Liability} &&&&&\\
Mack's CL model \cite{Mack} & 15,730&	11,526&	1,977&	-4,204&	213\%\\
RBNS CL method &11,494&	8,601&&		-2,893&\\
 \hline
\end{tabular}
\end{center}}
\caption{Mack's CL results on cumulative payments and CL RBNS reserves for both datasets of accident and liability insurance.}
\label{CL results}
\end{table}

In Table \ref{CL results} we report Mack's CL reserves and Mack's rooted mean squared error of prediction (RMSEP), which is a measure of the prediction uncertainty of the CL reserves. In case of accident insurance the CL reserves are 23,064 and the RMSEP is 1,663. Since in these examples we also know the lower triangle, we can benchmark these results against the true outstanding loss liabilities (OLL), i.e., the difference between the true ultimate claims $C_{i,J|\nu}$ and the payments $C_{i,I-i|\nu}$ already done at time $I$. In the case of accident insurance, the true OLL are 24,212 and the CL reserves underestimate the true liabilities by -1,148, which amounts to 69\% of the RMSEP. Thus, this underestimation is below one RMSEP, and it can well be explained by irreducible risk and parameter estimation uncertainty. We conclude that the CL method seems to work well for the accident insurance dataset.

For the liability insurance example, the situation is different. From Table \ref{CL results} we note that the CL method underestimates the true OLL by 213\% of Mack's RMSEP, in statistical terms one would reject the CL model in this case because the true outcome is not within two standard deviations (RMSEPs) of the prediction.

Mack's CL model \cite{Mack} gives a prediction for RBNS and IBNR claims because it does not distinguish w.r.t.~the claims reporting pattern. 
Using the method presented in Section \ref{Chain-ladder RBNS reserving}, we compute the RBNS reserves, we coin it `RBNS CL method' in Table \ref{CL results}. We observe that the prediction errors generally decrease (over-proportionally) compared to the decline in reserves, which indicates that part of the underestimation problem is due to IBNR claims.

\medskip

We will use these `RBNS CL method' results of Table \ref{CL results} to benchmark all further individual claims reserving methods (which only consider RBNS claims at time $I$).

\subsubsection{Accident insurance: Individual claims reserving results}
We now perform individual claims reserving using the FNN architectures 
$(\mu^{\rm FNN}_j)_{j=0}^3$ being fitted as described above.
The results are reported in Table \ref{FNN results accident}.

\begin{table}[h]
\centering
{\footnotesize
\begin{center}
\begin{tabular}{|c|r|rr|rr|rr|}
\hline
$i$ & True OLL & RBNS CL  & FNN & CL Error & FNN Error & CL Ind RMSE & FNN Ind RMSE \\
\hline\hline
1 &0	&0	&0	&0	&0	&0	&0\\
2&353&	339&	173&	-14	&-180&	1.499&	2.575\\
3 &1,017&	1,305&	1,262&	288&	246	&2.956&	2.954\\
4 &3,102	&3,099	&3,290&	-2	&189&	4.263&	4.248\\
5 &15,263&	14,216	&14,712	&-1,046&	-551&	8.240&	8.177\\\hline
Total &19,735	&18,959	&19,437	&-774&	-296&&\\		
 \hline
\end{tabular}
\end{center}}
\caption{Accident insurance: Results of individual claims prediction using the fitted FNN architectures $(\mu^{\rm FNN}_j)_{j=0}^3$.}
\label{FNN results accident}
\end{table}

The column `FNN' shows the individual claims prediction results using the fitted and balance-corrected FNN architectures $(\widehat{\mu}^{\rm bc-FNN}_j)_{j=0}^3$ and the inputs as described in Table \ref{tab:accidentdata}. Generally, we observe that these FNN predictions are closer to the true OLL than the RBNS CL forecasts (columns `CL Error' vs.~'FNN Error'), indicating that indeed one can effectively learn by including the additional features given in Table \ref{tab:accidentdata}. The total prediction error is -296 which is only 1.5\% of the total true OLL; this should also be compared to the RMSEP of 1,663 in Mack's model \cite{Mack}, see Table \ref{CL results}, indicating that we have an excellent forecast from the individual claims model.

The last two columns of Table \ref{FNN results accident} report the rooted mean squared errors (RMSE) between the forecasts (RBNS CL/FNN) and the true OLL on an individual claims level. We generally, observe a decreasing individual prediction error, except for accident year $i=2$, where we only predict one single period ahead. Thus, this individual claims reserving method leads to more accurate reserves on an individual claims level.
The decreased values still have a similar magnitude, which indicates that the dominating uncertainty component is irreducible risk (pure randomness) and the systematic effects that we can learn from the covariates live on a smaller scale -- low signal-to-noise ratio based on the available information -- which is common in many actuarial applications. This will further be discussed w.r.t.~the results of Table \ref{FNN results accident 2}, below.

\medskip

The critical part of recursive claims reserving methods is that such methods are prone to biases. E.g., if we overestimate accident period $i=2$, then this overestimation will propagate (in an alleviated manner) through the subsequent accident periods $i\ge 3$ through the recursive structure of the estimation procedure. For this reason, potential biases need careful consideration, and the fitting procedure may require regularization, as described in Section \ref{Bias control}. Our next goal is to check whether the bias control (balance correction) presented in \eqref{bc predictor} is effective.

Because in our example we know the true lower triangle, we can directly backtest for the bias problem -- note that in any reasonable real-world application this is not possible because the lower triangle is unknown (otherwise we would not need to predict it), but we perform this analysis here to verify that our proposal indeed works. The propagated bias problem comes from the fact that later periods use estimates of earlier periods in their fitting procedure, see \eqref{append 2}, or more generally in the learning data
${\cal L}_j$ we have responses and inputs, see \eqref{appended history},
\begin{equation}\label{append general}
\left(\widehat{C}_{i,J|\nu}, (C_{i,l|\nu}, \bX_{i,l|\nu})_{l=0}^{j}\right),
\end{equation}
for $I-(J-1) \le i \le I-(j+1)$ and $T_{i|\nu}\le j$.
This reflects the recursive estimation nature that has been used to compute the FNN results of Table \ref{FNN results accident}. To verify the effectiveness of our proposal \eqref{bc predictor} of the bias control, we use in a second fitting attempt the true ultimate claims which are available in our unusual set-up of knowing the lower triangle. That is, we replace in the learning data ${\cal L}_j$ the items \eqref{append general} during the fitting procedure by
\begin{equation}\label{append general true}
\left(C_{i,J|\nu}, (C_{i,l|\nu}, \bX_{i,l|\nu})_{l=0}^{j} \right).
\end{equation}
Because this no longer leads to a recursive estimation procedure (the responses are fully observed in this training procedure), this forecast process cannot lead to a propagating bias problem. We use the identical FNN architecture as above and we fit these FNNs using the learning data composed of \eqref{append general true}. The results are reported in Table \ref{FNN results accident 2}, and because \eqref{append general true} is not available in any real-world application, we set the corresponding results in square brackets.

\begin{table}[h]
\centering
{\footnotesize
\begin{center}
\begin{tabular}{|c|r|rr|rr|rr|}
\hline
$i$ & True OLL & RBNS CL  & FNN & CL Error & FNN Error & CL Ind RMSE & FNN Ind RMSE \\
\hline\hline
1 &0	&0	&0	&0	&0	&0	&0\\
2&353&	339&[173]&	-14	&[-180]&	1.499&	[2.575]\\
3 &1,017&	1,305&	[1,421]&	288&	[404]	&2.956&	[3.006]\\
4 &3,102	&3,099	&[3,183]&	-2	&[81]&	4.263&	[4.223]\\
5 &15,263&	14,216	&[14,581]	&-1,046&	[-682]&	8.240&	[8.212]\\\hline
Total &19,735	&18,959	&[19,358]	&-774&	[-377]&&\\		
 \hline
\end{tabular}
\end{center}}
\caption{Accident insurance: Results of individual claims prediction using the true ultimate claims for model fitting.}
\label{FNN results accident 2}
\end{table}

Comparing the individual claims reserving FNN results of Tables 
\ref{FNN results accident} and \ref{FNN results accident 2}, we observe a huge similarity, meaning that the information with estimated responses \eqref{append general} and the one with true responses \eqref{append general true} lead to almost the identical results, both concerning the total estimation error `FNN Error' against the true OLL as well as the errors on individual claims levels (last column `FNN Ind RMSE' of Tables \ref{FNN results accident} and \ref{FNN results accident 2}). This also verifies that the main uncertainty driver is irreducible risk (low signal-to-noise ratio based on the available information), and not missing precision of the ultimate claims prediction in \eqref{append general}. Of course, in a next step one may try to integrate the entire individual claims histories, collect more individual claims information, e.g., medical reports in accident insurance, and this may further decrease prediction error on an individual claims level.

\medskip

\begin{figure}[htb!]
\begin{center}
\begin{minipage}[t]{0.49\textwidth}
\begin{center}
\includegraphics[width=\textwidth]{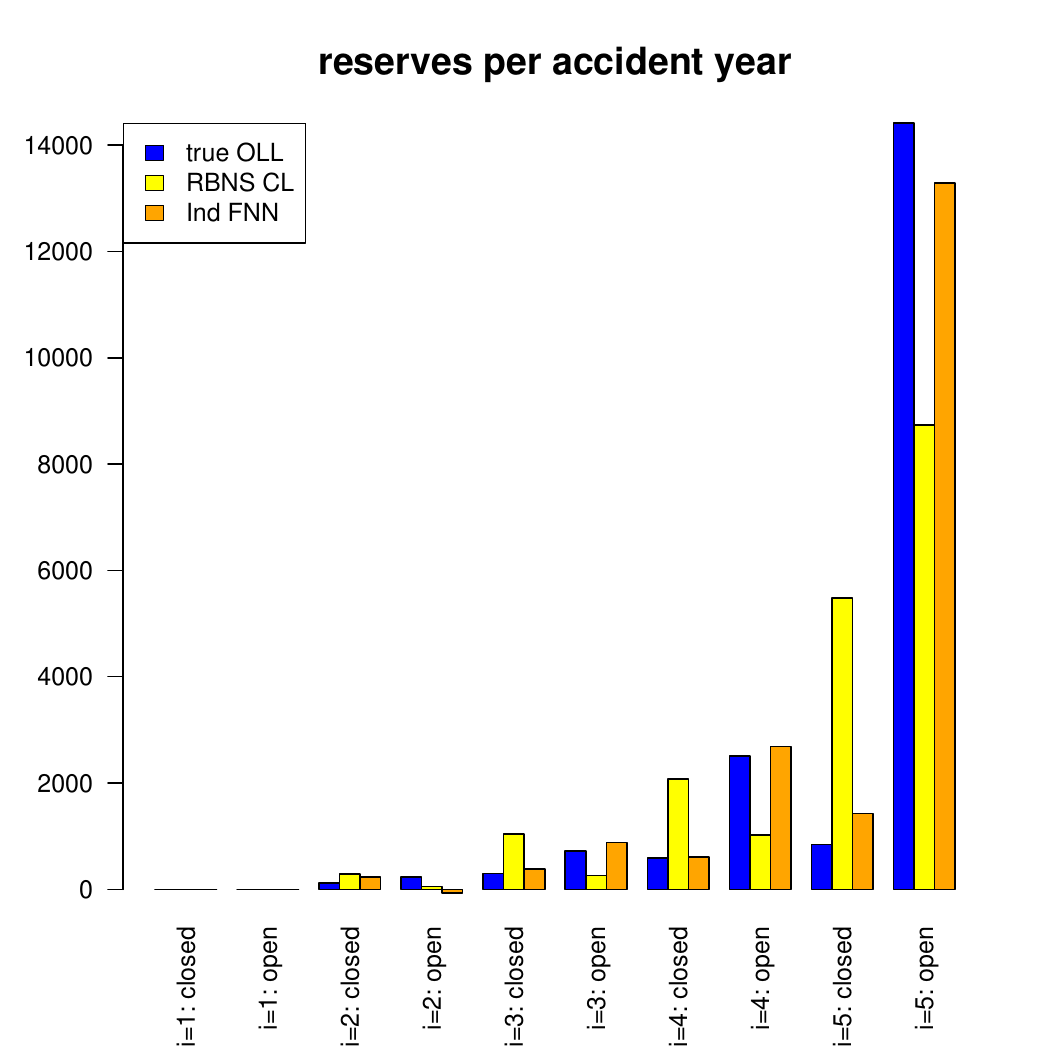}
\end{center}
\end{minipage}
\begin{minipage}[t]{0.49\textwidth}
\begin{center}
\includegraphics[width=\textwidth]{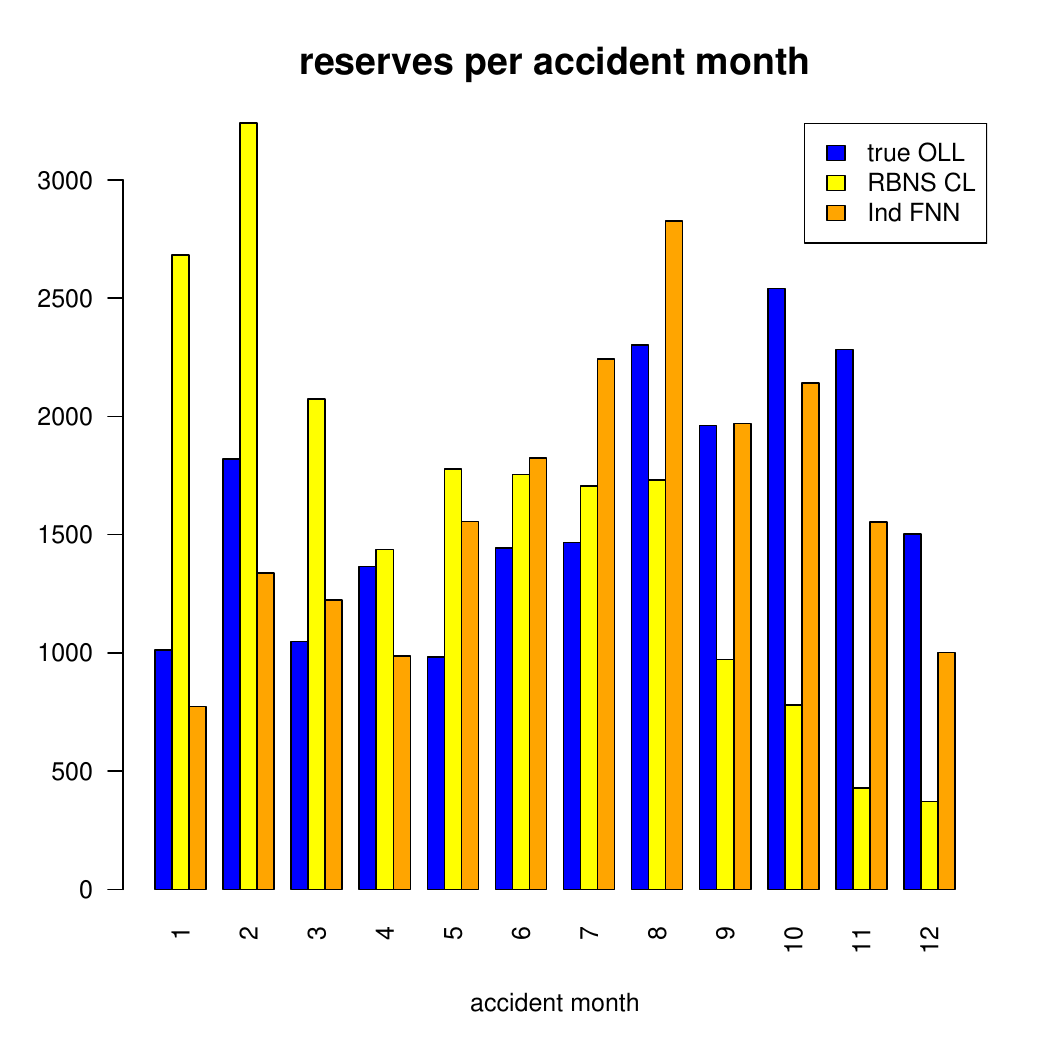}
\end{center}
\end{minipage}
\end{center}
\caption{(lhs) Reserves per accident year $i=1,\ldots, 5$ and separated by closed and open claims at time $I$, (rhs) per accident calendar month $cm \in \{1,\ldots, 12\}$. }
\label{figure reserves open claims}
\end{figure}

As a concluding analysis on the accident insurance dataset, we discuss the two plots shown in Figure \ref{figure reserves open claims}. The graph on the left-hand side shows the partition of the reserves w.r.t.~accident year $1\le i \le I=5$ and separated by closed and open claims at time $I$. The blue bars show the true OLLs, the yellow bars the CL predictions, and the orange bars the individual claims FNN forecasts. The individual FNN forecasts closely follow the true OLLs, thus, we seem to accurately capture the true OLLs, and we also correctly distinguish closed from open claims. It is also nice to see that indeed, there are quite some payments on closed claims (re-openings). The CL method, of course, cannot discriminate between closed and open claims, as the CL method does not consider a claim status label. 

The right-hand side of Figure \ref{figure reserves open claims} shows the reserves split by the accident month label. Also here we see a close alignment of individual claims FNN forecasts and true OLLs, generally having increasing reserves per accident month. This is explained by the fact that the payment data is considered on a calendar year grid, a January claim being more mature than a December claim, leading to generally lower reserves for early calendar months, but to receive the full picture, and should still complete this analysis by the different claim types in the different seasons; this information is not available here.

\subsubsection{General insurance: Individual claims reserving results}
Table \ref{FNN results liability} shows our second example of the liability insurance dataset. The interpretation of the results is rather similar to the accident insurance example. 

\begin{table}[h]
\centering
{\footnotesize
\begin{center}
\begin{tabular}{|c|r|rr|rr|rr|}
\hline
$i$ & True OLL & RBNS CL  & FNN & CL Error & FNN Error & CL Ind RMSE & FNN Ind RMSE \\
\hline\hline
1 &0	& 0	&0	&0	&0&	0&	0\\
2 & 361&	635&	761&	274&	400	&2.717	&4.483\\
3&3,233	&1,497	&1,743	&-1,736	&-1,491	&19.988&	18.574\\
4 &3,287&	2,488&	2,963 &	-799&	-324&	12.400&	11.704\\
5&4,613	&3,982&	4,038&	-631	&-575&	14.901	&13.907\\\hline
Total &11,494	&8,602&	9,505&	-2,892	&-1,894&&\\	
 \hline
\end{tabular}
\end{center}}
\caption{Liability insurance: Results of individual claims prediction using the fitted FNN architectures $(\mu^{\rm FNN}_j)_{j=0}^3$.}
\label{FNN results liability}
\end{table}

The reason for giving this second example is that additionally there is claims incurred $I_{i,j|\nu}$ available here. We use this claim incurred information as input in $\bX_{i,j|\nu}$, together with the resulting case reserves $R_{i,j|\nu}=I_{i,j|\nu}-C_{i,j|\nu}$. Interestingly, this has a very positive effect on the individual claims prediction error (last column of Table \ref{FNN results liability}), especially in the most recent accident year $i=5$. Intuitively it is clear that after the first development period $j=0$, the cumulative payments $C_{i,0|\nu}$ may not be very predictive for the ultimate claims prediction in this long-tailed business line. In this example the claims incurred estimates of the claims' adjusters are of good quality to improve individual claim prediction. Interestingly, for more mature accident years, there is less advantage in possessing claims incurred information. The reason may be two fold, either cumulative payments and claim status carry sufficient information in more developed years, or the quality of claims incurred is insufficient in more developed accident years (maybe due to a lack of continuous improvements/updates by claims' adjusters). 

Table \ref{FNN results liability 1} shows the results if we remove the claims incurred information from the inputs $\bX_{i,j|\nu}$. The biggest change in `FNN Ind RMSE' is indeed observed for the most recent accident year $i=5$, which supports that claims incurred is important information for the least developed claims.

\begin{table}[h]
\centering
{\footnotesize
\begin{center}
\begin{tabular}{|c|r|rr|rr|rr|}
\hline
$i$ & True OLL & RBNS CL  & FNN & CL Error & FNN Error & CL Ind Err & FNN Ind Error \\
\hline\hline
1 &0	& 0	&0	&0	&0&	0&	0\\
2 & 361&	635&503&	274&	142	&2.717	&5.072\\
3&3,233	&1,497	&1,664	&-1,736	&-1,569	&19.988&	20.079\\
4 &3,287&	2,488&	2,644 &	-799&	-643&	12.400&	12.026\\
5&4,613	&3,982&	5,028&	-631	&415&	14.901	&14.282\\\hline
Total &11,494	&8,602&	9,839&	-2,892	&-1,655&&\\	
 \hline
\end{tabular}
\end{center}}
\caption{Liability insurance: Results of individual claims prediction where we drop the claims incurred information from the inputs.}
\label{FNN results liability 1}
\end{table}

In view of Tables \ref{FNN results liability} and \ref{FNN results liability 1}, the bad performance of the RBNS CL method and partly of the individual claims reserving method can be traced back to accident year $i=3$. Analyzing this accident year in more detail, we find two individual claims that caused payments of 1,874 in development periods $j=3,4$, thus, two claims evolve to severe ones, which could not be anticipated by the systematic structure (covariates), but which needs to be attributed to the irreducible risk part. This also explains the large values in the two columns `Ind RSME' in Table \ref{FNN results liability}. Note that we considered all claims in our analysis -- in some literature large losses are discarded for a more stable fitting, see, e.g., Schneider--Schwab \cite[Section 3.1]{Schwab}; this removal simplifies modeling and analysis, but it is incorrect from a material point of view because these large claims are cost drivers that need to be borne by the insurer. Naturally, this large idiosyncratic part is not covered by the loss reserves (expected values), and it will trigger the solvency capital in a stand-alone view.


\section{Summary and next steps}
\label{Conclusions and Outlook}
We presented a novel way of computing the CL reserves by a recursive direct one-shot estimation and prediction procedure. This was achieved by restructuring the data and estimation procedure. This novel representation widely opens the door the ML applications on more granular individual claims data. In fact, this representation is suitable for any kind of input data, and due to the one-shot prediction approach it does not require to extrapolate the stochastic dynamic covariates, the latter being the main obstacle in most of the individual claims reserving methods. This barrier is now removed and we expect fast major developments in this field.

Another major advantage of our proposal over most other ML proposals is that we start from the familiar, well-proven chain-ladder method, the reserving actuaries are very familiar with. This allows one to easily use the chain-ladder predictions as guardrails to regularize the individual claims predictions staying close to the chain-ladder reserves. 

The main weakness of our proposal is the fact that it is a recursive procedure. Recursive methods are prone to propagating biases, which needs careful control in each step of the recursion. We presented the (simple) balance property approach for bias control, but being within a chain-ladder framework we can also envisage to control biases by classical chain-ladder predictions.

\begin{itemize}
\item Our main contribution is the recursive one-shot estimation and prediction procedure that enables the step from going from the classic chain-ladder method to individual claims reserving. Our (simple) example was rather meant as a proof of concept on a small data example using plain-vanilla neural networks. Naturally, there is huge potential (but also work to be done) to refine the machine learning method used to the specific individual claims reserving data selected, be it a neural network, a gradient boosting machine or other machine learning tools.
\item Naturally, the choice of the loss function plays a crucial role in model fitting. Our choice of the mean squared error  can clearly be improved so that the predictive model is improved over all ranges of the claim sizes, aiming at having accurate predictions both on large and small claims.
\item We have encountered some difficulties in modeling so-called zero claims, i.e., claims that can be closed without any payments. Our output activation function in the network architecture was the exponential one. This constrains predictions to strictly positive values. This point clearly needs more engineering to able to capture more effectively zero claims, especially those that have been closed already over several periods.
\item For our small scale example we imposed a Markov modeling assumption on the individual claims history, which implies that only the latest observation is relevant for prediction. Clearly this should be lifted to a model that includes the entire claims history. The straightforward next step is to employ a Transformer architecture including a CLS token for time-series data encoding. This allows for a data-driven decision about the validity of the Markov assumption.
\item For bias control we used probably the simplest method of balance correction. Clearly, there is more research necessary to prove its effectiveness and one should explore other methods.
\item At the moment, we fit a different network for every accident period considered. Clearly, we should ask for a more economic modeling, and hopefully different accident periods can share parts of the machine learning structure. However, it is not immediately clear how this can be achieved because of the fact that the algorithm is of a recursive nature that needs a careful bias control.
\item Our proposal considers a one-shot prediction of the ultimate claim. Naturally the same technology can also be used for a one-shot prediction of the entire future claim payment process. In machine learning jargon, this will require a sequence-to-sequence forecasting approach.
\item We only consider closed and reported but not settled (RBNS) claims. We still need to take care of incurred but not reported  (IBNR) claims. This will require a downstream model likely being based on a frequency-severity decomposition. For the frequency forecasting the same one-shot prediction approach could work, however, rather in a sequence-to-sequence forecasting structure because we care about the specific reporting pattern.
\end{itemize}

\bigskip

{\small 
\renewcommand{\baselinestretch}{.51}
}

\newpage

\appendix

\section{Proof}

{\Beweis
  {\bf Proof of Proposition \ref{CL prop}.}
  It suffices to prove that $\widehat{F}_{j}=\prod_{l=j}^{J-1}\widehat{f}_{l}$ for all $j \in \{0,\ldots, J-1\}$. The proof goes by induction. {\it Initialization.} For $J-1$, the claim follows from \eqref{CL2 A}.

{\it Induction step.} Now we consider the step $j+1 \to j \in \{0,\ldots, J-2\}$. Assume
$\widehat{F}_{k}=\prod_{l=k}^{J-1}\widehat{f}_{l}$ holds for all $k\in \{j+1, \ldots, J-1\}$.
We have
\begin{eqnarray*}
\widehat{F}_{j}&=& \frac{\sum_{i=1}^{I-j-1}\widehat{C}^*_{i,J}}
{\sum_{i=1}^{I-j-1}C_{i,j}}
~=~\frac{\sum_{i=1}^{I-j-1}C_{i,j+1}}
{\sum_{i=1}^{I-j-1}C_{i,j}}
\frac{\sum_{i=1}^{I-j-1}\widehat{C}^*_{i,J}}
{\sum_{i=1}^{I-j-1}C_{i,j+1}}
\\&=&
\widehat{f}_{j}~
\frac{\sum_{i=1}^{I-j-2}\widehat{C}^*_{i,J}+\widehat{C}^*_{I-j-1,J}}
      {\sum_{i=1}^{I-j-2}C_{i,j+1} + C_{I-j-1,j+1}}
  \\&=&
        \widehat{f}_{j}\left[
\frac{\sum_{i=1}^{I-j-2}\widehat{C}^*_{i,J}}
        {\sum_{i=1}^{I-j-2}C_{i,j+1} + C_{I-j-1,j+1}}
        +\frac{\widehat{C}^*_{I-j-1,J}}
        {\sum_{i=1}^{I-j-2}C_{i,j+1} + C_{I-j-1,j+1}}\right]
\\&=&
       \widehat{f}_{j}\left[
       \frac{\sum_{i=1}^{I-j-2}C_{i,j+1}}
        {\sum_{i=1}^{I-j-2}C_{i,j+1} + C_{I-j-1,j+1}}
       \frac{\sum_{i=1}^{I-j-2}\widehat{C}^*_{i,J}}
        {\sum_{i=1}^{I-j-2}C_{i,j+1} }       + 
        \frac{C_{I-j-1,j+1}}
             {\sum_{i=1}^{I-j-2}C_{i,j+1} + C_{I-j-1,j+1}}
      \frac{\widehat{C}^*_{I-j-1,J}}{ C_{I-j-1,j+1}}\right]
 \\&=&
       \widehat{f}_{j}\left[
       \frac{\sum_{i=1}^{I-j-2}C_{i,j+1}}
        {\sum_{i=1}^{I-j-2}C_{i,j+1} + C_{I-j-1,j+1}}\,\widehat{F}_{j+1}
       + 
        \frac{C_{I-j-1,j+1}}
             {\sum_{i=1}^{I-j-2}C_{i,j+1} + C_{I-j-1,j+1}}\,\widehat{F}_{j+1}\right]      
\\&=&
\widehat{f}_{j}\,
\widehat{F}_{j+1}
~=~\prod_{l=j}^{J-1}\widehat{f}_{l}.
\end{eqnarray*}
This completes the proof.
\EndProof}


\begin{thebibliography}{999}


 \bibitem{AvanziRL}
 Avanzi, B., Richman, R., Wong, B., W\"uthrich, M.V., Xie, Y. (2026).
 Reinforcement learning for micro-level claims reserving.
{\it arXiv}:2601.07637. 


\bibitem{Bladt}
  Bladt, M., Pittarello, G. (2025).
  Individual claims reserving using the Aalen–Johansen estimator.
  {\it ASTIN Bulletin - The Journal of the IAA} {\bf 55/1}, 29-49.


\bibitem{BF}
Bornhuetter, R.L., Ferguson, R.E. (1972). The actuary and IBNR.
{\it Proceedings CAS} {\bf 59}, 181-195.


\bibitem{Chaoubi}
Chaoubi, I., Besse, C., Cossette, H., C\^{o}t\'e, M.-P. (2023).
Micro-level reserving for general insurance claims using a long short-term memory network.
{\it Applied Stochastic Models in Business and Industry} {\bf 39/3}, 382-407.


\bibitem{DeFeliceMoriconi}
De Felice, M., Moriconi, F. (2019).
Claim watching and individual claims reserving using classification and regression trees.
{\it Risks} {\bf 7/4}, 102.


\bibitem{DelongLindholmW}
Delong, {\L}., Lindholm, M., W\"uthrich, M.V. (2022). 
Collective reserving using individual claims data.
{\it Scandinavian Actuarial Journal} {\bf 2022/1}, 1-28.


\bibitem{GabrielliEAJ}
Gabrielli, A. (2021).
An individual claims reserving model for reported claims.
{\it European Actuarial Journal} {\bf 11/2}, 541-577.

\bibitem{Hiabu}
  Hiabu, M., Hofman, E.D., Pitarello, G. (2023).
  A machine learning approach based on survival analysis for IBNR frequencies in non-life reserving.
  {\it arXiv}:2312.14549.
  
\bibitem{Kuo1}
Kuo, K. (2019).
DeepTriangle: a deep learning approach to loss reserving.
{\it Risks} {\bf 7/3}, article 97.

\bibitem{Kuo2}
Kuo, K. (2020).
Individual claims forecasting with Bayesian mixture density networks.
{\it arXiv}:2003.02453v1.

  
\bibitem{Lopez2}
Lopez, O., Milhaud, X. (2021).
Individual reserving and nonparametric estimation of claim amounts subject to large reporting delays. 
{\it Scandinavian Actuarial Journal} {\bf 2021/1}, 34-53.

\bibitem{Lopez}
Lopez, O., Milhaud, X., Th\'{e}rond, P.-E.  (2019). 
A tree-based algorithm adapted to microlevel reserving and long development claims.
{\it ASTIN Bulletin - The Journal of the IAA} {\bf 49/3}, 741-762.

\bibitem{LorenzSchmidt}
Lorenz, H., Schmidt, K.D. (2016).
Grossing up method.
In: {\it Handbook on Loss Reserving},
Radtke, M., Schmidt, K.D., Schnaus, A. (eds.), Springer, 127-131.


\bibitem{Mack}
Mack, T. (1993). Distribution-free calculation of the standard error of chain ladder reserve estimates. 
{\it ASTIN Bulletin - The Journal of the IAA} {\bf 23/2}, 213-225.

\bibitem{RichmanCred}
Richman, R., Scognamiglio, S.,  W{\"u}thrich, M.V. (2025).
The credibility transformer.
{\it European Actuarial Journal} {\bf  15/2}, 345-379.

\bibitem{RichmanW2}
Richman, R., W\"uthrich, M.V. (2020).
Nagging predictors.
{\it Risks} {\bf 8/3}, article 83.


\bibitem{Rosenlund}
Rosenlund, S. (2012). 
Bootstrapping individual claim histories. 
{\it ASTIN Bulletin - The Journal of the IAA} {\bf 42/1}, 291-324. 


\bibitem{Schwab}
Schneider, J.C., Schwab, B. (2025).
Advancing loss reserving: a hybrid neural network approach for individual claim development prediction.
{\it Journal of Risk and Insurance} {\bf 92/2}, 389-423.

\bibitem{Schnieper}
Schnieper, R. (1991).
Separating true IBNR from IBNER claims.
{\it ASTIN Bulletin - The Journal of the IAA} {\bf 21/1}, 111-127.


\bibitem{Turcotte}
Turcotte, R., Shi, P. (2026).
Individual loss reserving for multi-coverage insurance.
{\it ASTIN Bulletin - The Journal of the IAA} {\bf 56/2}, to appear.



\bibitem{WTree}
W\"uthrich, M.V. (2018).
Machine learning in individual claims reserving.
{\it Scandinavian Actuarial Journal} {\bf 2018/6}, 465-480. 

\bibitem{Wbalance}
W\"uthrich, M.V.  (2020).
Bias regularization in neural network models for general insurance pricing.
{\it European Actuarial Journal} {\bf 10/1}, 179-202.


\bibitem{WM2008}
W\"uthrich, M.V., Merz, M. (2008).
{\it Stochastic Claims Reserving Methods in Insurance}. Wiley.



\bibitem{WM2023}
W\"uthrich, M.V., Merz, M. (2023).
{\it Statistical Foundations of Actuarial Learning and its Applications.}
Springer.


\end{thebibliography}
\end{document}